\documentclass[journal,10pt]{IEEEtran}

\usepackage{blindtext, graphicx,booktabs}
\usepackage{cite}
\usepackage{amsmath}
\usepackage{amssymb}
\newtheorem{lemma}{Lemma}
\newtheorem{theorem}{Theorem}
\newcommand{\ignore}[1]{}
\begin{document}
\title{Multi-Class Management with Sub-Class Service for Autonomous Electric Mobility On-Demand Systems }
\author{\IEEEauthorblockN{Syrine Belakaria$^*$, Mustafa Ammous$^*$, Sameh Sorour$^*$ and Ahmed Abdel-Rahim$^\dag$$^\ddag$}\\
\IEEEauthorblockA{$^*$Department of Electrical and Computer Engineering,
University of Idaho,
Moscow, ID, USA\\ $^\dag$Department of Civil and Environmental Engineering,
University of Idaho,
Moscow, ID, USA\\$^\ddag$National Institute for Advanced Transportation Technologies, University of Idaho, Moscow, ID, USA\\
Email: \{bela7898, ammo1375\}@vandals.uidaho.edu, \{samehsorour, ahmed\}@uidaho.edu}
\thanks{This study was supported by the Pacific Northwest University Transportation center (Pac trans) project KLK864. The authors would like to thank all who contribute to this study from both funding agencies.}
}

\maketitle

\begin{abstract}
Despite the significant advances in vehicle automation and electrification, the next-decade aspirations for massive deployments of autonomous electric mobility on demand (AEMoD) services are still threatened by two major bottlenecks, namely the computational and charging delays. This paper proposes a solution for these two challenges by suggesting the use of fog computing for AEMoD systems, and developing an optimized charging scheme for its vehicles with and multi-class dispatching scheme for the customers. A queuing model representing the proposed multi-class management scheme with sub-class service is first introduced. The stability conditions of the system in a given city zone are then derived. Decisions on the proportions of each class vehicles to partially/fully charge, or directly serve customers of possible sub-classes are then optimized in order to minimize the maximum response time of the system. Results show the merits of our optimized model compared to a previously proposed scheme and other non-optimized policies. 
\end{abstract}

\begin{IEEEkeywords}
Autonomous Mobility On-Demand; Electric Vehicle; Fog-based Architecture; Dispatching; Charging; Queuing Systems.
\end{IEEEkeywords}
\IEEEpeerreviewmaketitle
\section{Introduction}

Urban transportation systems are facing tremendous challenges nowadays due to the dominant dependency and massive increases on private vehicle ownership, which result in dramatic increases in road congestion, parking demand \cite{ref5}, \cite{ref6}, and carbon footprint \cite{ref3} \cite{ref4}. These challenges can be mitigated with the significant advances of vehicle electrification, automation, and connectivity. With more than 10 million self-driving cars expected to be on the road by 2025 \cite{ref11}, it is forecasted that vehicle ownership will significantly decline by 2025, as it will be replaced by the novel concept of Autonomous Electric Mobility on-Demand (AEMoD) services \cite{ref9,ref10}. In such system, customers will simply need to press some buttons on an app to promptly get an autonomous electric vehicle transporting them door-to-door, with no pick-up/drop-off and driving responsibilities, no dedicated parking needs, no vehicle insurance and maintenance costs, and extra in-vehicle leisure times. With these qualities, AEMoD systems will succeed in attracting millions of subscribers and providing hassle-free private urban mobility.\\ 
\indent Despite the great aspirations for wide AEMoD service deployments by early-to-mid next decade, the timeliness (and thus success) of such service is threatened by two major bottlenecks. First, the expected massive demand of AEMoD services will result in excessive if not prohibitive computational and communication delays if cloud based approaches are employed for the micro-operation of such systems. Moreover, the typical full-battery charging rates of electric vehicles will not be able to cope with the gigantic numbers of vehicles involved in these systems, thus resulting in instabilities and unbounded customer delays. Several recent works Recent works have addressed important problems in AMoD systems by building different operation models for them like a distributed spatially averaged queuing model and a lumped Jackson network model \cite{ref1}  also the system was \cite{ref16} cast into a closed multi-class BCMP queuing network to  solve the routing problem on congested roads. Many key factors were not considered in these works in order to simplify the mathematical resolution. None of these papers considered the computational architecture for massive demands on such services, the vehicle electrification, and the influence of charging limitations on its stability.  \\
\ignore{
\begin{figure}[t]
\centering
      \includegraphics[width=.48\textwidth]{fog1.pdf}
      \caption{Fog-based architecture for AEMoD system operation}
      \label{fog}
 \end{figure}}
In \cite{ref17} \cite{ref18} we proposed a closely related model management model for Amods Systems. We proposed to resolve the first limitation, communication/computation delays, by suggesting the exploitation of the new and trendy \emph{fog-based networking and computing architectures} \cite{ref28}. The privileges brought by this technique \cite{ref17}, will allow handling instantaneous decision making applications such as AEMoD system operations in a distributed and accelerated way. The fog controller in each service zone is responsible of collecting information about customer requests, vehicle in-flow to the service zone, their state-of-charge (SoC), and the available full-battery charging rates in the service zone. Given the collected information, it can promptly make dispatching, and charging decisions for these vehicles in a timely manner. \\ 
\indent In order to solve the second problem, we proposed previously \cite{ref17,ref18} that the fog controller will smartly cope with the available charging capabilities of each service zone, by assigning to each customer a vehicle that has enough charge to serve him without in route charging. In the former solution, arriving vehicles in each service zone are subdivided into different classes in ascending order of their SoC corresponding to the different customer classes. Different proportions of each class vehicles will either wait (without charging) for dispatching to its corresponding customer class or partially charge to serve a customer from the same class. Vehicles arriving with depleted batteries will be allowed to either partially or fully charge. Despite the valuable results given by this model but the dispatching process of the system can result on having all the vehicles depleted by the end of the service which may cause the instability of the system. In this paper we keep the same charging scheme and we propose an enhanced dispatching process that allow each vehicle to serve all the sub-classes of customers (customers that needs lower state of charge) \\
\indent The question now is: \emph{To maintain charging stability and minimize the maximum response time of the system, What are the optimal Charging and sub-classes dispatching decisions?} To address this question, a queuing model representing the proposed multi-class management with sub-class service scheme is first introduced. The stability conditions of this model. Decisions on the proportions of each class vehicles to partially/fully charge, or directly serve customers and decision on which class will be served are then optimized. Finally, the merits of our proposed optimized decision scheme are tested and compared to several non optimized schemes.
\section{System Model}
We consider one service zone controlled by a fog controller connected to: (1) the service request apps of customers in the zone; (2) the AEMoD vehicles; (3) $C$ rapid charging points distributed in the service zone and designed for short-term partial charging; and (4) one spacious rapid charging station designed for long-term full charging. AEMoD vehicles enter the service in this zone after dropping off their latest customers in it. Their detection as free vehicles by the zone's controller can thus be modeled as a Poisson process with rate $\lambda_v $. Customers request service from the system according to a Poisson process. Both customers and vehicles are classified into $n$ classes based on an ascending order of their required trip distance and the corresponding SoC to cover this distance, respectively. From the thinning property of Poisson processes, the arrival process of Class $i$ customers and vehicles, $i \in \{0,\dots, n\}$, are both independent Poisson processes with rates $\lambda_c^{(i)}$ and $\lambda_v p_i$, where $p_i$ is the probability that the SoC of an arriving vehicle to the system belongs to Class $i$. Note that $p_0$ is the probability that a vehicle arrive with a depleted battery, and is thus not able to serve immediately. Consequently, $\lambda^{(0)}_c = 0$ as no customer will request a vehicle that cannot travel any distance. On the other hand, $p_n$ is also equal to 0, because no vehicle can arrive to the system fully charged as it has just finished a prior trip.\\
\indent Upon arrival, each vehicle of Class $i$, $i\in\{1,\dots,n-1\}$, will park anywhere in the zone until it is called by the fog controller to either: (1) join vehicles that will serve customer with their current state of charge with probability $q_i$ (The served customer can be from any Sub-class $j$ with $j \leq i$); or (2) partially charge up to the SoC of class $i+1$ at any of the $C$ charging points (whenever any of them becomes free), with probability $\overline{q}_i=1-q_i$, before parking again in waiting to serve a customer from any Sub-class $j$ with $j \leq i+1$. As for Class $0$ vehicles that are incapable of serving before charging, they will be directed to either fully charge at the central charging station with probability $q_0$, or partially charge at one of $C$ charging points with probability $\overline{q}_0 = 1-q_0$. In the former and latter cases, the vehicle after charging will wait to serve customers of any Sub-class $j \leq n$ and $1$, respectively.\\
Considering the above explanation, Each vehicle, whether decided to serve immediately or decided to charge before serving, will be able to serve: (1) customers from same class with probability $\Pi_{ii}$; or (2) customers with a trip distance from any Sub-class with probability $\Pi_{ij}$.
\begin{figure}[t]
\centering
 \includegraphics[width=.5\textwidth]{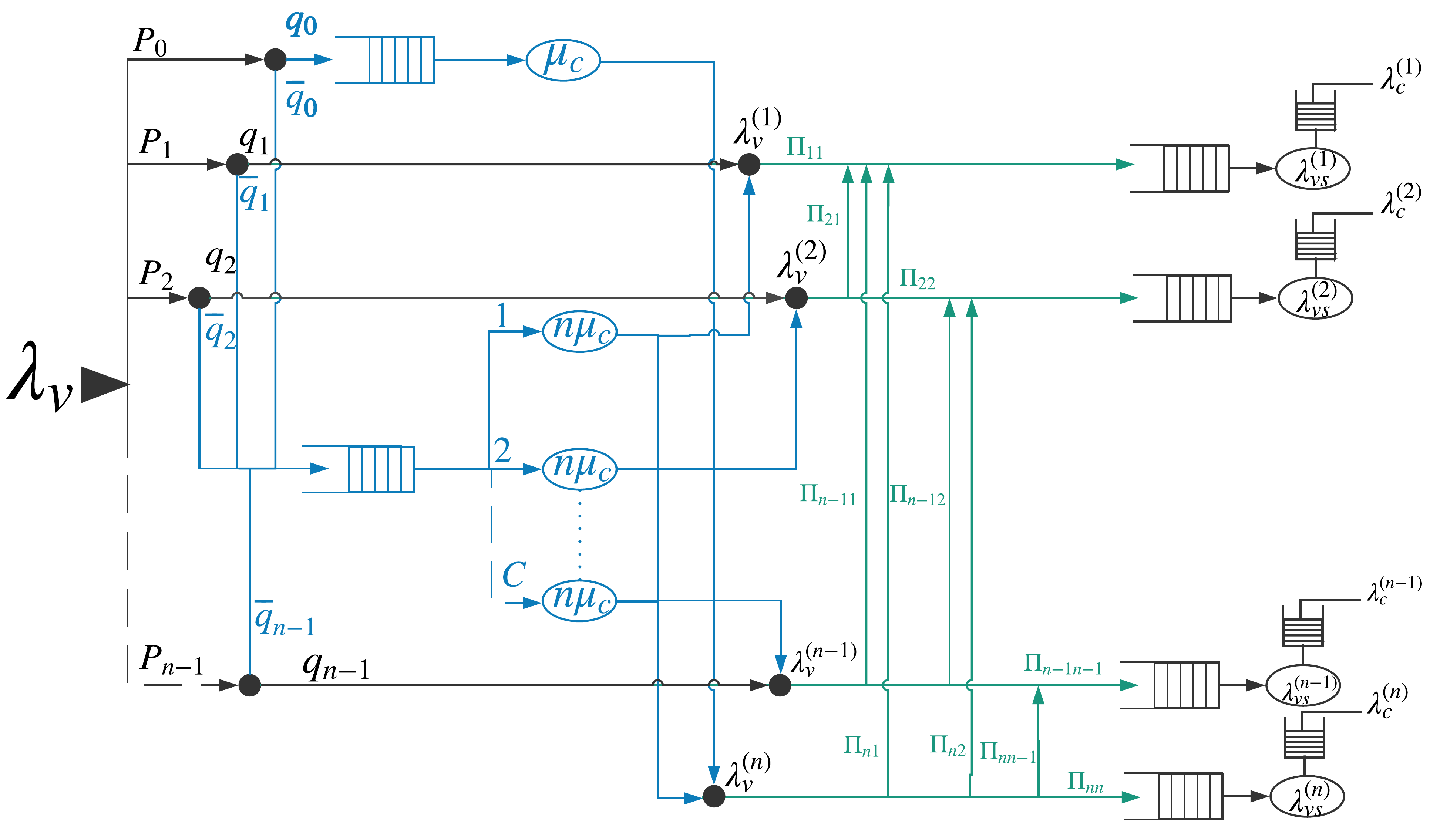}
\caption{Joint dispatching and partially/fully charging model, abstracting an AEMoD system in one service zone.}\label{fig:model}
    \end{figure}
\indent As widely used in the literature (e.g., \cite{ref15,ref16}), the full charging time of a vehicle with a depleted battery is assumed to be exponentially distributed with rate $\mu_c$. Given uniform SoC quantization among the $n$ vehicle classes, the partial charging time can then be modeled as an exponential random variable with rate $n\mu_c$. Note that the larger rate of the partial charging process is not due to a speed-up in the charging process but rather due to the reduced time of partially charging. The customers belonging to Class $i$, arriving at rate $\lambda_c^{(i)}$, will be served at a rate of $\lambda_{vs}^{(i)}$, which includes summation of proportions of arrival rates of vehicles that: (1) arrived to the zone with a SoC belonging to Class $j ~ \forall j \geq i$ and were directed to wait to serve Class $i ~\forall i \leq j$ customers; or (2) arrived to the zone with a SoC belonging to Class $j-1$ and were directed to partially charge to be able to serve a Sub-Class $i ~ \forall i\leq j$ customers. Given the above description and modeling of variables, the entire zone dynamics can thus be modeled by the queuing system \ignore{depicted in Fig.\ref{fig:model}}. This system includes $n$ M/M/1 queues for the $n$ classes of customer service, one M/M/1 queue for the central charging station, and one M/M/C queue representing the partial charging process at the $C$ charging points.\\
\indent Assuming that the service zones will be designed to guarantee a maximum time for a vehicle to reach a customer, our goal in this paper is to minimize the maximum expected response time of the entire system. By response time, we mean the time needed that vehicle starts moving from its parking or charging spot towards this customer. 
\section{System stability conditions}
In this section, we first deduce the stability conditions of our proposed joint dispatching and charging system, using the basic laws of queuing theory. 
Each class of vehicles with an arrival rate $ \lambda_{v}^{(i)} $ will be characterized by its SoC when it is ready to serve customers. 
Each of the $n$ classes of customers are served by a separate queue of vehicles, with $ \lambda_{vs}^{(i)} $ being the arrival rate of the vehicles that are available to serve the customers of the $i^{th}$ class.  Consequently, it is the service rate of the customers $i^{th}$ arrival queues.
We can thus deduce from \ignore{Fig. \ref{fig:model} and} the system model in the previous section the rate of vehicles with SoC that allows to serve a class $i$ or any sub-class $j \leq i$ that requires lower SoC to serve its customers:
\begin{equation}\label{eq:1}
\begin{aligned}
& & &\lambda_v^{(i)} = \lambda_v(p_{i-1}\overline{q}_{i-1} + p_{i}q_{i}) , \; i = 1, \ldots, n-1.\\
& & &\lambda_v^{(n)} = \lambda_v(p_{n-1}\overline{q}_{n-1} + p_{0}q_{0})\\
\end{aligned}
\end{equation}
Since we know that $\overline{q}_{i} + q_{i} = 1$
Then we substitute $\overline{q}_{i}$ by  $1 - q_{i}$ in order to have a system with $n$ variables
\begin{equation}\label{eq:2}
\begin{aligned}
& & &\lambda_v^{(i)} = \lambda_v(p_{i-1} - p_{i-1}{q_{i-1}} + p_{i}q_{i}) , \; i = 1, \ldots, n-1\\
& & &\lambda_v^{(n)} = \lambda_v(p_{n-1} - p_{n-1}{q_{n-1}} + p_{0}q_{0}) \ignore{, \; i = n}
\end{aligned}
\end{equation}
We can also deduce the expression of the rate of vehicles that will actually serve a class of customers $i$:  
\begin{equation}\label{eq:3}
\begin{aligned}
& & & \lambda_{vs}^{(i)} =\sum_{k=i}^n \lambda_v^{(k)} \Pi_{ki} , \; i = 1, \ldots, n 
\end{aligned}
\end{equation}
By injecting the expression of $\lambda_v^k$ in (\ref{eq:2}) in (\ref{eq:3}), we find: 
\begin{equation}\label{eq:4}
\begin{aligned}
& & & \lambda_{vs}^{(i)} =\lambda_v \sum_{k=i}^{n-1} (p_{k-1} - p_{k-1}{q_{k-1}} + p_{k}q_{k})\Pi_{ki} \\
& & & + \lambda_v (p_{n-1} - p_{n-1}{q_{n-1}} + p_{0}q_{0})\Pi_{ni} , \; i = 1, \ldots, n-1 \\
& & & \lambda_{vs}^{(n)} = \lambda_v (p_{n-1} - p_{n-1}{q_{n-1}} + p_{0}q_{0})\Pi_{nn}
\end{aligned}
\end{equation}
From the well-known stability condition of an M/M/1 queue: 
\begin{equation}\label{eq:5}
\begin{aligned}
& & & \lambda_{vs}^{(i)} > \lambda_c^{(i)} , \; i = 1, \ldots, n 
\end{aligned}
\end{equation}
\ignore{
It is also established from M/M/1 queue analysis that the average response (i.e., service) time for any customer in the $i$-th class can be expressed a:
\begin{equation} \label{eq:6}
\dfrac {1}{\lambda_{vs}^{(i)}-\lambda_c^{(i)}}
\end{equation}}
To guarantee customers' satisfaction, the fog controller of each zone must impose an average response time limit $T$ for any class. We can thus express this average response time constraint for the customers of the $i$-th class as:
\begin{equation} \label{eq:7}
\dfrac {1}{\lambda_{vs}^{(i)}-\lambda_c^{(i)}} \leq T \text{ or } \lambda_{vs}^{(i)} - \lambda_c^{(i)}\geq R, \; \text{with } R = \frac{1}{T}
\end{equation}
\ignore{
which can also be re-written as:
\begin{equation} \label{eq:8}
\lambda_{vs}^{(i)} - \lambda_c^{(i)}\geq R, \; \text{with } R = \frac{1}{T}
\end{equation}}
\indent Before reaching the customer service queues, the vehicles will go through a decision step of either to go to these queues immediately or partially charge. From the system model, we have the following stability constraints on the $C$ charging points and central charging station queues, respectively: 
\begin{equation}\label{eq:9}
\begin{aligned}
& & & \sum ^{n-1}_{i=0}\lambda_v(p_{i} - p_{i}{q_{i}}) < C (n \mu_c) \\
& & &\lambda_v p_{0}q_{0} < \mu_c
\end{aligned}
\end{equation}
The following lemma allows the estimation of the average needed vehicles arrival for a given service zone.
\begin{lemma}\label{lem1}
For the entire zone stability, and fulfillment of the average response time limit for all its classes, the average vehicles arrival rate must be lower bounded by:
\begin{equation}\label{eq:10}
\begin{aligned}
\lambda_v \geq \sum ^{n}_{i=1}{\lambda_c^{(i)}} + n R
\end{aligned}
\end{equation}
\end{lemma}
\begin{proof}
The proof of Lemma \ref{lem1} is in Appendix A of \cite{ref13}.
\end{proof}

\section{Joint Charging and Dispatching optimization}
\subsection{Problem Formulation}
The goal of this paper is to minimize the maximum expected response time of the system's classes. The response time of any class is defined as the average of the duration from any customer request until a vehicle is dispatched to serve him/her. The maximum expected response time is expressed as: 
\begin{equation} \label{eq:12}
\max_{i\in\{1,\dots,n\}} \left\{\dfrac {1}{\lambda_{vs}^{(i)}-\lambda_c^{(i)}}\right\}
\end{equation}
It is obvious that the system's class having the maximum expected response time is the one that have the minimum expected response rate. In other words, we have:
\begin{equation}\label{eq:13}
\arg\max_{i\in\{1,\dots,n\}} \left\{\dfrac {1}{\lambda_{vs}^{(i)}-\lambda_c^{(i)}}\right\} = \arg\min_{i\in\{1,\dots,n\}} \left\{\lambda_{vs}^{(i)}-\lambda_c^{(i)}\right\}
\end{equation}
Consequently, minimizing the maximum expected response time is equivalent to maximizing the minimum expected response rate. Using the epigraph form \cite{refconvxbook} of the latter problem, we get the following stochastic optimization problem:
\begin{subequations}\label{eq:14}
\begin{align}
&\qquad\qquad \underset{q_0,\ldots ,q_{n-1},\Pi_{11},\ldots,\Pi_{nn}}{\text{maximize }} R \\
\text{s.t}  & \nonumber\\
& \lambda_v \sum_{k=i}^{n-1} (p_{k-1} - p_{k-1}{q_{k-1}} + p_{k}q_{k})\Pi_{ki}\\
& + \lambda_v (p_{n-1} - p_{n-1}{q_{n-1}} + p_{0}q_{0})\Pi_{ni} - \lambda_c^{(i)} \geq R\\
& \qquad\qquad \qquad\qquad \qquad\qquad  i = 1, \ldots, n-1 \label{eq:14-C1}\\
& \lambda_v (p_{n-1} - p_{n-1}{q_{n-1}} + p_{0}q_{0})\Pi_{nn} - \lambda^{(n)}_c \geq R \label{eq:14-C2} \\
&\sum ^{n-1}_{i=0}\lambda_v(p_{i} - p_{i}{q_{i}}) < C (n \mu_c) \label{eq:14-C3}\\
&\lambda_v p_{0}q_{0} < \mu_c  \label{eq:14-C4} \\
& \sum ^{i}_{j=1} \Pi_{ij}= 1 , \; i = 1, \ldots, n  \label{eq:14-C5}\\
& 0 \leq  \Pi_{ij} \leq  1 ,\; i = 1, \ldots, n,\; j = 1, \ldots, i \label{eq:14-C6}\\
& 0 \leq q_{i} \leq 1  , \; i = 0, \ldots, n-1  \label{eq:14-C7}\\
& \sum ^{n-1}_{i=0}p_i = 1 , \; 0 \leq  p_{i} \leq  1  , \; i = 0, \ldots, n-1  \label{eq:14-C9}\\
& 0 < R \leq \frac{\lambda_v - \sum ^{n}_{i=1}{\lambda_c^{(i)}}}{n}  \label{eq:14-C8}\\
\end{align}
\end{subequations}
The $n$ constraints in (\ref{eq:14-C1}) and (\ref{eq:14-C2}) represent the epigraph form's constraints on the original objective function in the right hand side of (\ref{eq:13}), after separation \cite{refconvxbook} and substituting every $\lambda_{vs}^{(i)}$ by its expansion form in (\ref{eq:4}). The constraints in (\ref{eq:14-C3}) and (\ref{eq:14-C4}) represent the stability conditions on charging queues. The constraints in (\ref{eq:14-C5}), (\ref{eq:14-C6}),(\ref{eq:14-C7}) and (\ref{eq:14-C9}) are the axiomatic constraints on the probabilities (i.e., values being between 0 and 1, and sum equal to 1). The Finally, Constraint (\ref{eq:14-C8}) is a positivity constraint on the minimum expected response rate. 
Finally, Constraint (\ref{eq:14-C8}) is is a positivity constraint and the upper bound on $R$ introduced by Lemma (\ref{lem1}).\\

\subsection{Lower Bound Analytical Solutions}
The optimization problem in (\ref{eq:14}) is a quadratic non-convex problem with second order differentiable objective and constraint functions. Usually, the solution obtained by using the Lagrangian and KKT analysis for such non-convex problems provides a lower bound on the actual optimal solution. Consequently, we propose to solve the above problem by first finding the solution derived through Lagrangian and KKT analysis, then, if needed, iteratively tightening this solution to the feasibility set of the original problem.
The Lagrangian function associated with the optimization problem in (\ref{eq:14}) is given by the following expression:
\begin{multline}\label{eq:15}
\begin{aligned}
& L(R,\mathbf{q},\boldsymbol{\Pi},\boldsymbol{\alpha},\boldsymbol{\beta},\boldsymbol{\gamma},\boldsymbol{\omega},\boldsymbol{\mu},\boldsymbol{\nu},\boldsymbol{\delta}) = - \sum_{i=0}^{n-1}\omega_{i}q_{i} - \omega_{n}(R-\epsilon_2)\\
 &+\sum_{i=1}^{n-1}\alpha_{i}[\lambda_c^{(i)} - \lambda_v \sum_{k=i}^{n-1} (p_{k-1} - p_{k-1}{q_{k-1}} + p_{k}q_{k})\Pi_{ki} \\
& \qquad\qquad \quad- \lambda_v (p_{n-1} - p_{n-1}{q_{n-1}} + p_{0}q_{0})\Pi_{ni}+ R]\\
& + \alpha_{n} (\lambda_c^{(n)}-\lambda_v (p_{n-1} - p_{n-1}{q_{n-1}} + p_{0}q_{0})\Pi_{ni} +R )\\
& + \beta_{0} (\sum ^{n-1}_{i=0}\lambda_v(p_{i} - p_{i}{q_{i}}) - C (n \mu_c) ) + \beta_{1}(\lambda_v p_{0}q_{0} - \mu_c)\\
&+ \sum_{i=0}^{n-1} \gamma_{i}(q_{i} - 1) + \gamma_{n}(R - \frac{\lambda_v - \sum ^{n}_{i=1}{\lambda_c^{(i)}}}{n}) - R\\
& + \sum_{i=1}^{n} \sum_{j=1}^{i} \nu_{ij} (\Pi_{ij}-1)- \mu_{ij}\Pi_{ij} +\sum_{i=1}^{n} \delta_i(\sum_{k=1}^{i}\Pi_{ik}-1)
\end{aligned}
\end{multline}
where:
\begin{itemize}
\item $\mathbf{q}= [q_0,\dots,q_{n-1}]$ is the vector of charing decisions.
\item $\mathbf{\Pi}= [\Pi_{ij}]$ is the vector of dispatching decisions to serve customers.
\item $\boldsymbol{\alpha} = [\alpha_{i}]$, such that $\alpha_{i}$ is the associated Lagrange multiplier to the $i$-th customer queues inequality.
\item $\boldsymbol{\beta} = [\beta_{i}]$, such that $\beta{i}$ is the associated Lagrange multiplier to the $i$-th  charging queues inequality.
\item $\boldsymbol{\delta} = [\delta_{i}]$, such that $\delta_{i}$ is the  associated Lagrange multiplier to the $i$-th equality constraint on the dispatching decision.
\item $\boldsymbol{\gamma} = [\gamma_{i}]$, such that $\gamma_{i}$ is the associated Lagrange multiplier to the $i$-th  upper bound inequality on the charging decisions and the expected response time.
\item $\boldsymbol{\omega} = [\omega_{i}]$, such that $\omega_{i}$ is the  associated Lagrange multiplier to the $i$-th lower bound inequality on the charging decisions and the expected response time.
\item $\boldsymbol{\mu} = [\mu_{ij}]$, such that $\mu_{ij}$ is the  associated Lagrange multiplier to the $j$-th lower bound inequality on the dispatching decision $\Pi_{ij}$.
\item $\boldsymbol{\nu} = [\nu_{ij}]$, such that $\nu_{ij}$ is the  associated Lagrange multiplier to the $j$-th upper bound inequality on the dispatching decision $\Pi_{ij}$.

\end{itemize}
For more accurate resolutions, Three small positive constants $\epsilon_0$, $\epsilon_1$ and $\epsilon_2$ are added to the stability conditions on the charging queues and the positivity condition on the maximum expected waiting time to make them non strict inequalities.\\
\indent Solving the equations given by the KKT conditions on the problem equality and inequality constraints, the following theorem illustrates the optimal lower bound solutions of the problem in (\ref{eq:14}).
\begin{theorem}\label{thm1}
The lower bound solution of the optimization problem in (\ref{eq:14}), obtained from Lagrangian and KKT analysis can be expressed as follows:
\begin{figure*}
\begin{equation}\label{eq:16}
\begin{aligned}
& R^* = \begin{cases}
           \frac{\lambda_v - \sum ^{n}_{i=1}{\lambda_c^{(i)}}}{n} &  \gamma_n^* \ne 0 \\    
            \epsilon_2 & \omega_n^* \ne 0\\ 
            \sum_{i=1}^{n-1}\alpha_{i}^*( \lambda_v \sum_{k=i}^{n-1} (p_{k-1} - p_{k-1}{q_{k-1}^*} + p_{k}q_{k}^*)\Pi_{ki}^*+ \lambda_v (p_{n-1} - p_{n-1}{q_{n-1}^*} + p_{0}q_{0}^*)\Pi_{ni}^* - \lambda_c^{(i)})\\
            + \alpha_{n}^* ( \lambda_v (p_{n-1} - p_{n-1}{q_{n-1}^*} + p_{0}q_{0}^*)\Pi_{nn}^* - \lambda^{(n)}_c  ) & Otherwise\\
            \end{cases} \\ 
& q_{0}^* = \begin{cases}
            0 &  \alpha_{1}^*\Pi_{11}^* - \sum_{i=1}^n \alpha_{i}^* \Pi_{ni}^* - \beta_0^* + \beta_1^*>0 \\
            1 & \alpha_{1}^*\Pi_{11}^* - \sum_{i=1}^n \alpha_{i}^* \Pi_{ni}^* - \beta_0^* + \beta_1^* < 0 \\
            \frac{ \lambda_c^{(n)}+\lambda_v p_{n-1}q_{n-1}^*\Pi_{nn}^* - \lambda_v p_{n-1} \Pi_{nn}^* + R^*}{\lambda_v p_{0} \Pi{nn}^* } & \alpha_{n}^* \ne 0 \\
            \frac{\mu_c}{\lambda_v^* p_0}& \beta_1^* \ne 0 \\
            \zeta_0(R^*,q^*,\Pi^*,\alpha^*,\beta^*,\gamma^*,\omega^*,\mu^*,\nu^*,\delta^*)& Otherwise \\
            \end{cases} \\ 
& q_{i}^* =   \begin{cases}
            0 \qquad\qquad\qquad\qquad\qquad\qquad\qquad\qquad\quad  \alpha_{i+1}^* - \alpha_{i}^* -\beta_0^* > 0  \\
            1 \qquad\qquad\qquad\qquad\qquad\qquad\qquad\qquad\quad  \alpha_{i+1}^* - \alpha_{i}^* - \beta_0^*<0     \qquad\qquad\qquad\qquad i= 1, \ldots, n-1.\\
            \frac{ R^*+\lambda_c^{(i)}+\lambda_v\left[\sum_{k=i+2}^{n-1}(p_{k-1}{q_{k-1}^*} -p_{k}{q_{k}^*})\Pi_{ki} + (p_{n-1}{q_{n-1}^*} -p_{0}{q_{0}^*})\Pi_{ni} - \sum_{k=i}^{n}p_{k-1}\Pi_{ki}+p_{i-1}q_{i-1}^*\Pi_{ii}^*-p_{i+1}q_{i+1}\Pi_{i+1i+1} \right]}{\lambda_v p_{i}(\Pi_{ii}-\Pi_{i+1i+1}) } &\alpha_{i}^* \ne 0\\
             \zeta_i(R^*,q^*,\Pi^*,\alpha^*,\beta^*,\gamma^*,\omega^*,\mu^*,\nu^*,\delta^*)& Otherwise \\
            \end{cases}    \\
            & \Pi_{ij}^* =   \begin{cases}
            0 &  \alpha_{j}^*\lambda_v (p_{i-1}q_{i-1^*}-p_{i-1}-p_{i}q_{i}^*) + \delta_i^* > 0 \\
            1 &  \alpha_{j}^*\lambda_v (p_{i-1}q_{i-1^*}-p_{i-1}-p_{i}q_{i}^*) + \delta_i^*  < 0  \\     \zeta_{ij}(R^*,q^*,\Pi^*,\alpha^*,\beta^*,\gamma^*,\omega^*,\mu^*,\nu^*,\delta^*)& Otherwise \\
            \end{cases}    \quad i= 1, \ldots, n-1.
\end{aligned}
\end{equation}
\hrule
\end{figure*}
where $\zeta_i$ and $\zeta_{ij}$ are the solution that that maximize $\underset{\mathbf{q},\mathbf{\Pi}}\inf~ L(\mathbf{q},\Pi^*,R^*,\alpha^*,\beta^*,\gamma^*,\omega^*,\mu^*,\nu^*,\delta^*)$
\end{theorem}
\begin{proof}
The proof of Theorem \ref{thm1} is in Appendix C in \cite{ref13}.
\end{proof}
\subsection{Solution Tightening}
As stated earlier, the closed-form solution derived in the previous section from analyzing the constraints' KKT conditions does not always match with the optimal solution of the original optimization problem, and is sometimes a non-feasible lower bound on our problem. Unfortunately, there is no method to find the exact closed-from solution of non-convex optimization. However, starting from the derived lower bound, we can numerically tighten this solution by iterating toward the feasible set of the original problem. There are several algorithms to iteratively tighten lower bound solutions, one of which is the \textit{Suggest-and-Improve algorithm} algorithm proposed in \cite{qcqp} to tighten non-convex quadratic problems. We will thus propose to employ this method whenever the KKT conditions based solution is not feasible and tightening is required.

\section{Simulation Results}
In this section, we test the merits of our proposed scheme using extensive simulations. The metric used to evaluate these merits is the maximum expected response times of the different classes. For all the performed simulation figures, the full-charging rate of a vehicle is set to $\mu_c = 0.033$ mins$^{-1}$, and the number of charging points $C=40$.\\
\indent Fig. \ref{fig:comp-decreasing} depicts the maximum expected response time for different values of $\sum ^{n}_{i=1}{\lambda_c^{(i)}}$, while fixing $\lambda_v$ to 8 min$^{-1}$. For this setting, $n=7$ is the smallest number of classes that satisfy the stability condition in Lemma 2 in \cite{ref17}. From queuing theory rules\cite{Probabook}\cite{Probabook2} the more serving queues a system have, the higher the waiting time will be. Moreover, in previous related work \cite{ref17} \cite{ref18}, we showed that increasing the number of classes $n$ beyond its strict lower bound introduced in Lemma 2 in \cite{ref17} will damage the system performance and increase the maximum response time.\\
\indent Fig. \ref{fig:comp-decreasing} compare the maximum expected response time performances against $\sum ^{n}_{i=1}{\lambda_c^{(i)}}$, for different decision approaches namely our derived optimal decisions to the following decisions sets:
\begin{enumerate}
\item Optimized charging decisions (i.e. $q_i~\forall~i$) with same class dispatching (i.e. $\Pi_{ii} = 1~\forall~i$ and $\Pi_{ij}=0~\forall~i,j\neq i$)
\item Always partially charge decisions (i.e. $q_i=0~\forall~i$) with same class dispatching (i.e. $\Pi_{ii} = 1~\forall~i$ and $\Pi_{ij}=0~\forall~i,j\neq i$)
\item Equal split charging decisions (i.e. $q_i=0.5~\forall~i$) with same class dispatching (i.e. $\Pi_{ii} = 1~\forall~i$ and $\Pi_{ij}=0~\forall~i,j\neq i$)
\item Always partially charge decisions (i.e. $q_i=0~\forall~i$) with proportional sub-classes dispatching decisions (i.e. $\Pi_{ij}$ proportional to the customers sub-classes needs)
\item Equal split charge decisions (i.e. $q_i=0~\forall~i$) with proportional sub-classes dispatching decisions (i.e. $\Pi_{ij}$ proportional to the customers sub-classes needs)
\end{enumerate}
These five schemes represent the possible non-optimized policies, in which each vehicle takes its own fixed decision irrespective of the system parameters. These schemes are possible in case of a non connected and optimized system.\\
\indent Fig. \ref{fig:comp-decreasing} compares these approaches with a decreasing SoC distribution. The figure clearly show superior performances for our derived optimal policy compared to the other policies, especially as $\sum ^{n}_{i=1}{\lambda_c^{(i)}}$ gets closer to $\lambda_v$, which are the most properly engineers scenarios (as large differences between these two quantities results in very low utilization), This approves the expression found in lemma \ref{lem1}. A Gains of 49.3\%, 69.8\%, 93.22\%, 86.7\% and 94.4\% in the performances, can be noticed compared to the previously stated policies respectively.
\ignore{\indent For Gaussian distribution (Fig.\ref{fig:comp-Gaussian}) the results are very close to the first policy (which is also the previously optimized policy proposed in  \cite{ref17} \cite{ref18}). For the second policy the gain is up to 16.6\% while compared to the third, forth and fifth previously stated policies, our model has 88\%, 90\% and 92\% gain respectively.\\
\indent For Decreasing distribution (Fig.\ref{fig:comp-decreasing}) a Gains of 49.3\%, 69.8\%, 93.22\%, 86.7\% and 94.4\% in the performances, can be noticed compared to the previously stated policies respectively.}\\
\indent Fig. \ref{fig:comp-charging} shows the study of the resilience requirements for our considered model in the critical scenarios of sudden reduction in the number of charging sources within the zone. This reduction may occur due to either natural (e.g., typical failures of one or more stations) or intentional (e.g., a malicious attack on the fog controller blocking its access to these sources). The resilience measure that the fog controller can take in these scenarios is to notify its customers of a transient increase in the vehicles' response times given the available vehicles in the zone. For this, we are only comparing the new proposed model to our previously proposed model. The figures shows clearly the advantage brought by the sub-class dispatching model. The gain gets higher in critical scenarios and reaches up to 65\% with very acceptable maximum response time even is very low energy resources.
This demonstrates the importance of our proposed scheme in achieving better customer satisfaction.\\

\begin{figure}[t]
\centering
  \includegraphics[width=0.85\linewidth]{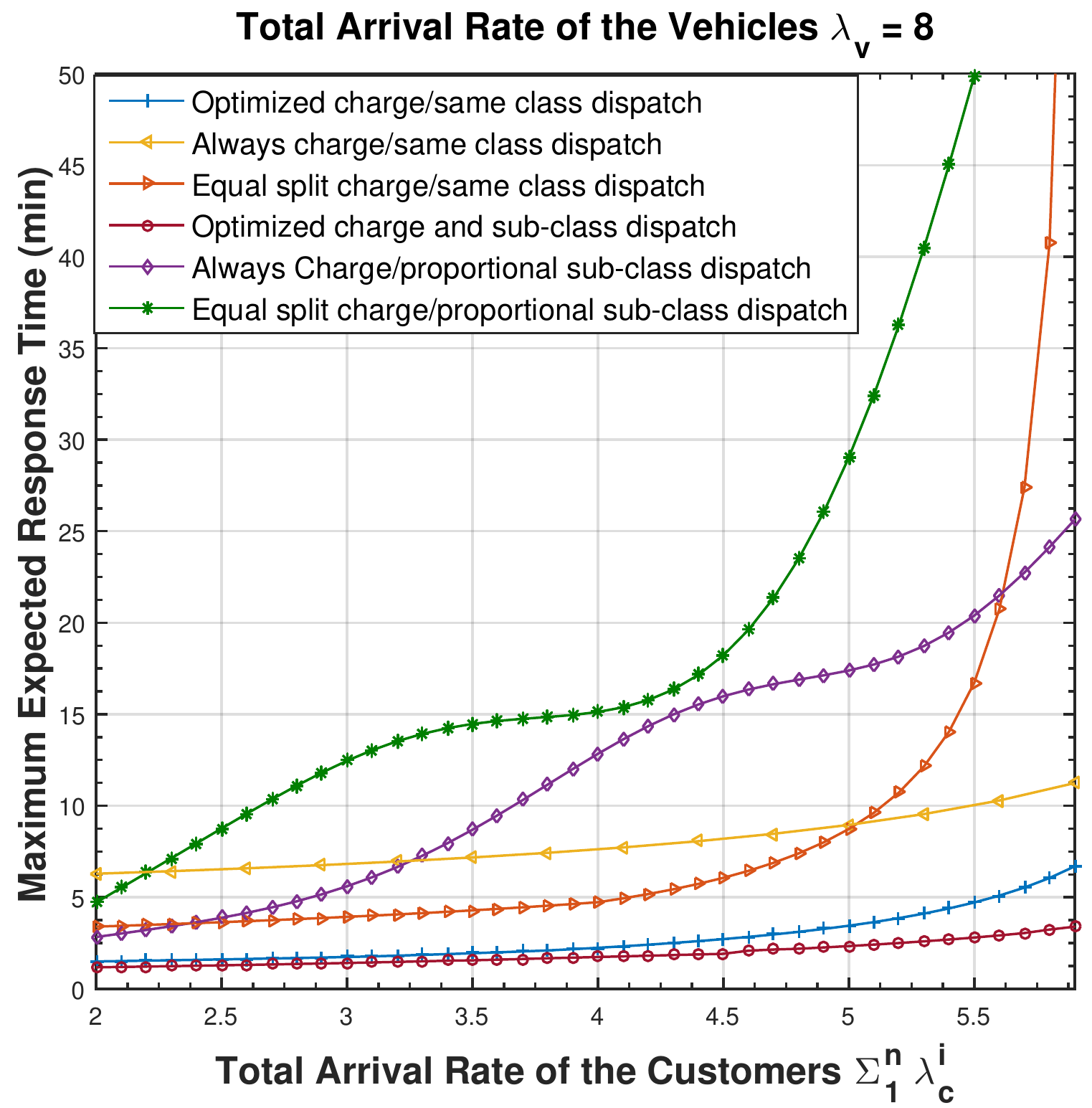}
    \caption{Comparison to non-optimized policies for Decreasing SoC distribution. \label{fig:comp-decreasing} }
\end{figure}
\begin{figure}[t]
\centering
  \includegraphics[width=0.85\linewidth]{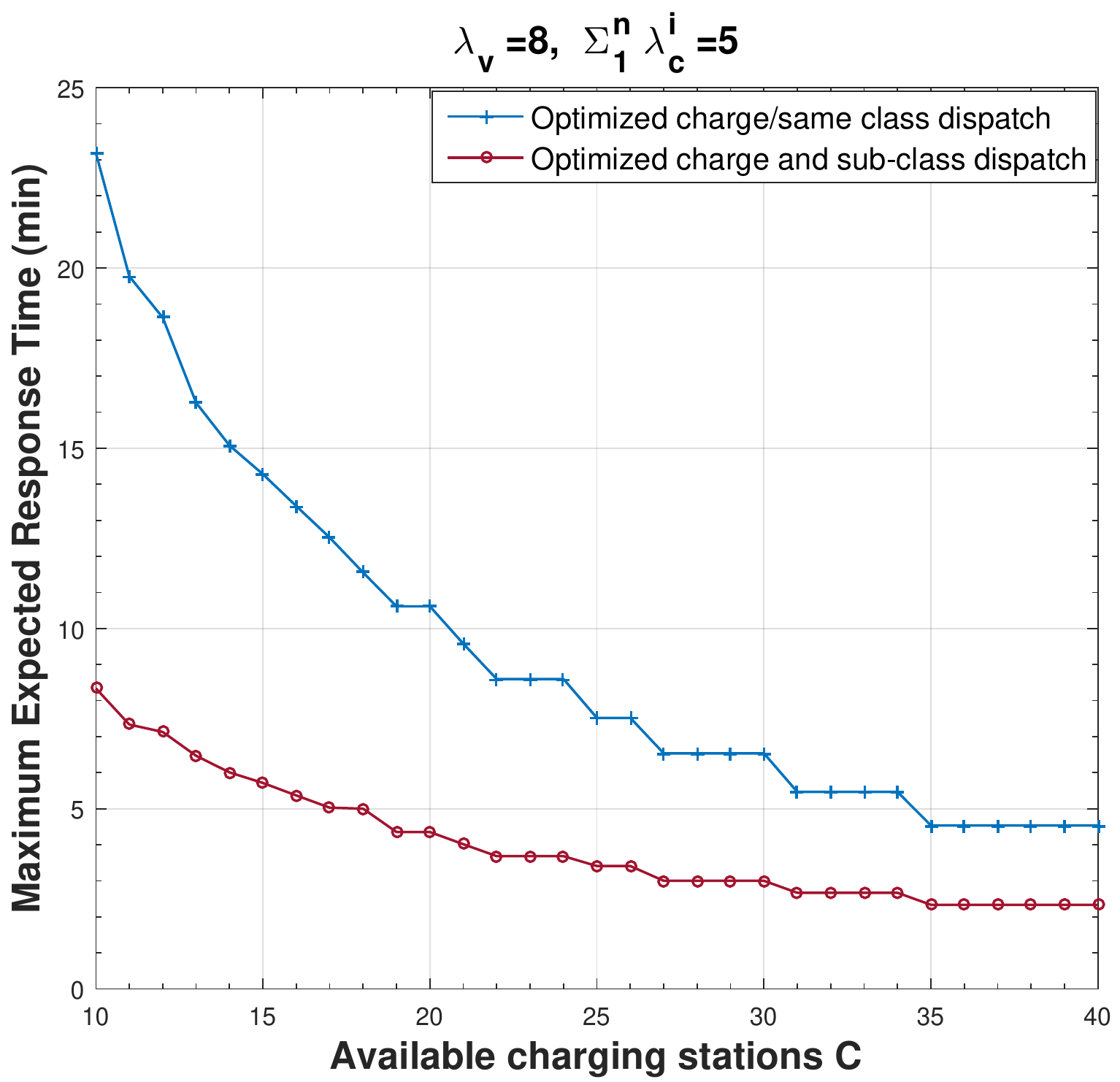}
    \caption{Effect of varying charging points availability. \label{fig:comp-charging} }
\end{figure}
\section{Conclusion}
In this paper, we proposed solutions to the computational and charging bottlenecks threatening the success of AEMoD systems. The computational bottleneck can be resolved by employing a fog-based architecture to distribute the optimization loads over different service zones, reduce communication delays, and matches the nature of dispatching and charging processes of AEMoD vehicles. We also proposed a multi-class dispatching and charging scheme and developed its queuing model and stability conditions. We then formulated the problem of optimizing the proportions of vehicles of each class that will partially/fully charge or directly serve customers of same class or any lower sub-class as an optimization problem, in order to minimize the maximum expected system response time while respecting the system stability constraints. The optimal decisions and corresponding maximum response time were analytically derived. Simulation results demonstrated both the merits of our proposed optimal decision scheme compared to typical non-optimized schemes and previously optimized scheme, and its performance for different distributions of vehicle SoC and customer trip distances.
\ignore{
For the future work, we are planning to modify the system formulation and compare the final results.This will include using the optimization over the service load of the queues $\frac{\mu}{\lambda}$ or over the a whole vector $\dfrac{1}{\lambda_v^({i})-\lambda_c^({i})}$.This work will be extended to time-varying optimization and predictive control of the next time frame that will be applied to real life collected data.
The future optimization will include the number of cars depending on the demand and the area.
}

\newpage
\onecolumn
\appendices
\section{Proof of Lemma \ref{lem1}}\label{app:lem1}
From (\ref{eq:4}) and (\ref{eq:7}) we have : 
\begin{equation}\label{eq:18}
\begin{aligned}
& & & \lambda_c^{(i)} + R \leq \lambda_v \sum_{k=i}^{n-1} (p_{k-1} - p_{k-1}{q_{k-1}} + p_{k}q_{k})\Pi_{ki} + \lambda_v (p_{n-1} - p_{n-1}{q_{n-1}} + p_{0}q_{0})\Pi_{ni} , \; i = 1, \ldots, n-1 \\
& & & \lambda_c^{(n)} +R \leq \lambda_v (p_{n-1} - p_{n-1}{q_{n-1}} + p_{0}q_{0})\Pi_{nn}
\end{aligned}
\end{equation}
The summation of all the inequalities in (\ref{eq:18}) gives a new inequality 
\begin{equation}\label{eq:19}
\begin{aligned}
\sum ^{n}_{i=1}{\lambda_c^{(i)}} + n R \leq \sum ^{n}_{i=1}\sum ^{n}_{k=1} \lambda_v^{(k)} \Pi_{ki}
\end{aligned}
\end{equation}
Which is equivalent to: 
\begin{equation}\label{eq:20}
\begin{aligned}
\sum ^{n}_{i=1}{\lambda_c^{(i)}} + n R \leq \sum ^{n}_{i=1}\lambda_v^{(i)}\sum ^{i}_{k=1}  \Pi_{ik}
\end{aligned}
\end{equation}
Since We have $\sum ^{i}_{j=1} \Pi_{ij}= 1 , \; i = 1, \ldots, n $ then: 
\begin{equation}\label{eq:21}
\begin{aligned}
\sum^{n}_{i=1}{\lambda_c^{(i)}} +nR \leq \sum ^{n}_{i=1}\lambda_v^{(i)}
\end{aligned}
\end{equation}
From (\ref{eq:1}) and (\ref{eq:2})
\begin{equation}\label{eq:22}
\begin{aligned}
\sum ^{n}_{i=1}\lambda_v^{(i)}= \sum ^{n-1}_{i=1}{(p_{i-1}\overline{q}_{i-1} + p_{i}q_{i})} + (p_{n-1}\overline{q}_{n-1} + p_{0}q_{0})
\end{aligned}
\end{equation}
Since ${\overline{q}_{i} + q_{i}}=1$ and we have $\sum ^{n-1}_{i=0}{p_{i}} = 1  $  then:
\begin{equation}\label{eq:23}
\begin{aligned}
\sum ^{n}_{i=1}{\lambda_c^{(i)}} +nR \leq \lambda_v
\end{aligned}
\end{equation}

\section{Proof of Theorem \ref{thm1}} \label{app:thm1}
Applying the KKT conditions to the inequalities constraints of (\ref{eq:14}), we get:
\begin{equation}\label{eq:28}
\begin{aligned}
& & &\alpha_{i}^*( \lambda_c^{(i)} + R^* - \lambda_v \sum_{k=i}^{n-1} (p_{k-1} - p_{k-1}{q_{k-1}^*} + p_{k}q_{k}^*)\Pi_{ki}^* - \lambda_v (p_{n-1} - p_{n-1}{q_{n-1}^*} + p_{0}q_{0}^*)\Pi_{ni}^* )=0 \; i = 1, \ldots, n-1 \\
& & &\alpha_{n}^* ( \lambda^{(n)}_c + R^* - \lambda_v (p_{n-1} - p_{n-1}{q_{n-1}^*} + p_{0}q_{0}^*)\Pi_{nn}^*  ) = 0.\\
& & &\beta_{0}^* (\sum ^{n-1}_{i=0}\lambda_v(p_{i} - p_{i}{q_{i}^*}) - C (n \mu_c)) = 0.\\
& & &\beta_{1}^*(\lambda_v p_{0}q_{0}^* - \mu_c) = 0 \\
& & & \gamma_{i}^*(q_{i}^* - 1) = 0  , \;i = 0, \ldots, n-1.\\
& & & \gamma_{n}^*(R^* - \frac{\lambda_v - \sum ^{n}_{i=1}{\lambda_c^{(i)}}}{n} ) = 0 .\\
& & &\omega_{i}^*q_{i}^* = 0 , \;i = 0, \ldots, n-1.\\
& & & \omega_{n}^*(R^*-\epsilon_2) = 0.\\
& & & \nu_{ij}^* (\Pi_{ij}^*-1)=0, \;i = 1, \ldots, n , \;j = 1, \ldots i.\\
& & & \mu_{ij}^*\Pi_{ij}^* =0 , \;i = 1, \ldots, n , \;j = 1, \ldots i.\\
& & & \delta_i^*(\sum_{j=1}^{i}\Pi_{ij}^*-1)=0, \;i = 1, \ldots, n .
\end{aligned}
\end{equation}
Likewise, applying the KKT conditions to the Lagrangian function in (\ref{eq:15}), and knowing that the gradient of the Lagrangian function goes to $0$ at the lower bound solution, we get the following set of equalities:
\begin{equation}\label{eq:29}
\begin{aligned}
& \frac{\partial L}{\partial q_i} =\lambda_v p_{i}\left(\sum_{j=1}^{i} \alpha_{j}^*(\Pi_{i+1j}^*-\Pi_{ij}^*) +\alpha_{i+1}^*\Pi_{i+1i+1}^* - \beta_0^*\right) -\omega_{i}^* + \gamma_{i}^*=0 , \;  i= 1, \ldots, n-1.\\
& \frac{\partial L}{\partial q_0} = \lambda_v p_{0}\left(\alpha_{1}^*\Pi_{11}^* - \sum_{j=1}^{n}\alpha_{j}^*\Pi_{nj}^*- \beta_0^* + \beta_1^*\right) - \omega_{0}^* + \gamma_{0}^* =0 \ \\
& \frac{\partial L}{\partial \Pi_{ij}} =\alpha_{j}^*\lambda_v \left(p_{i-1}{q_{i-1}}^* - p_{i-1} - p_{i}q_{i}^*\right) +\delta_i^* +\nu_{ij}^*-\mu_{ij}^* = 0\\
&\frac{\partial L}{\partial \Pi_{nj}} = \alpha_{j}^*\lambda_v \left(p_{n-1}{q_{n-1}}^* - p_{n-1} - p_{0}q_{0}^*\right) +\delta_n^* +\nu_{nj}^*-\mu_{nj}^*= 0\\
&\frac{\partial L}{\partial R}= -1 + \sum_{i=1}^{n} \alpha_{i}^* - \omega_n^* + \gamma_n^*= 0
\end{aligned}
\end{equation}
Multiplying the each of the partial derivatives in (\ref{eq:29}) by the derivation variable itself combined with the KKT conditions of the variables lower bounds inequalities given by (\ref{eq:28}) gives :
\begin{equation}\label{eq:30}
\begin{aligned}
& \frac{\partial L}{\partial q_i}\times q_i =q_i^*\lambda_v p_{i}\left(\sum_{j=1}^{i} \alpha_{j}^*(\Pi_{i+1j}^*-\Pi_{ij}^*) +\alpha_{i+1}^*\Pi_{i+1i+1}^* - \beta_0^*\right) + \gamma_{i}^*=0 , \;  i= 1, \ldots, n-1.\\
& \frac{\partial L}{\partial q_0}\times q_0 =q_0^* \lambda_v p_{0}\left(\alpha_{1}^*\Pi_{11}^* - \sum_{j=1}^{n}\alpha_{j}^*\Pi_{nj}^*- \beta_0^* + \beta_1^*\right) + \gamma_{0}^* =0 \ \\
& \frac{\partial L}{\partial \Pi_{ij}} \times \Pi_{ij} =\Pi_{ij}^* \left(\alpha_{j}^*\lambda_v \left(p_{i-1}{q_{i-1}}^* - p_{i-1} - p_{i}q_{i}^*\right) +\delta_i^*\right) +\nu_{ij}^* = 0\\
&\frac{\partial L}{\partial \Pi_{nj}} \times \Pi_{nj} = \Pi_{nj}^* \left(\alpha_{j}^*\lambda_v \left(p_{n-1}{q_{n-1}}^* - p_{n-1} - p_{0}q_{0}^*\right) +\delta_n^* \right)+\nu_{nj}^*= 0\\
&\frac{\partial L}{\partial R}\times R= -R + R \sum_{i=1}^{n} \alpha_{i}^* - \omega_n^* \epsilon_2 + \gamma_n^* (\frac{\lambda_v - \sum ^{n}_{i=1}{\lambda_c^{(i)}}}{n})= 0
\end{aligned}
\end{equation}
When we inject the result of the first four equations in (\ref{eq:30}) in the KKT conditions on the upper bound conditions of the variables $q_i$ and $\Pi_{ij}$ we find: 
\begin{equation}\label{eq:31}
\begin{aligned}
& q_i^*(q_i^*-1)\left(\sum_{j=1}^{i} \alpha_{j}^*(\Pi_{i+1j}^*-\Pi_{ij}^*) +\alpha_{i+1}^*\Pi_{i+1i+1}^* - \beta_0^*\right)=0 , \;  i= 1, \ldots, n-1.\\
& q_0^*(q_0^*-1)\left(\alpha_{1}^*\Pi_{11}^* - \sum_{j=1}^{n}\alpha_{j}^*\Pi_{nj}^*- \beta_0^* + \beta_1^*\right) =0 \ \\
& \Pi_{ij}^*(\Pi_{ij}^*-1) \left(\alpha_{j}^*\lambda_v \left(p_{i-1}{q_{i-1}}^* - p_{i-1} - p_{i}q_{i}^*\right) +\delta_i^*\right) = 0\\
& \Pi_{nj}^*(\Pi_{nj}^*-1) \left(\alpha_{j}^*\lambda_v \left(p_{n-1}{q_{n-1}}^* - p_{n-1} - p_{0}q_{0}^*\right) +\delta_n^* \right)= 0\\
\end{aligned}
\end{equation}
From (\ref{eq:31}) we have :\\
$0 < q_{0}^* <1$ only if $\alpha_{1}^*\Pi_{11}^* - \sum_{j=1}^{n}\alpha_{j}^*\Pi_{nj}^*- \beta_0^* + \beta_1^*=0$\\
$0 < q_{i}^* <1$ only if $\sum_{j=1}^{i} \alpha_{j}^*(\Pi_{i+1j}^*-\Pi_{ij}^*) +\alpha_{i+1}^*\Pi_{i+1i+1}^* - \beta_0^*=0$ \\
$0 < \Pi_{ij}^* <1$ only if $\alpha_{j}^*\lambda_v \left(p_{i-1}{q_{i-1}}^* - p_{i-1} - p_{i}q_{i}^*\right) +\delta_i^*=0$ \\
$0 < \Pi_{nj}^* <1$ only if $\alpha_{j}^*\lambda_v \left(p_{n-1}{q_{n-1}}^* - p_{n-1} - p_{0}q_{0}^*\right) +\delta_n^*=0$ \\

Since $0 \leq q_{i}^* \leq 1$ and $0 \leq \Pi_{ij}^*  \leq 1$ then these equalities may not always be true \\
if  $\alpha_{1}^*\Pi_{11}^* - \sum_{j=1}^{n}\alpha_{j}^*\Pi_{nj}^*- \beta_0^* + \beta_1^*>0$ and we know that $\gamma_{0}^* \geq 0$ then  $\gamma_{0}^* = 0$ which gives $q_{0}^* \ne 1 $ and $q_{0}^* = 0 $.\\
if  $\sum_{j=1}^{i} \alpha_{j}^*(\Pi_{i+1j}^*-\Pi_{ij}^*) +\alpha_{i+1}^*\Pi_{i+1i+1}^* - \beta_0^*>0$ which gives $q_{i}^* \ne 1 $ and $q_{i}^* = 0 $\\
if  $\alpha_{1}^*\Pi_{11}^* - \sum_{j=1}^{n}\alpha_{j}^*\Pi_{nj}^*- \beta_0^* + \beta_1^*<0$ then  $\gamma_{0}^* > 0$ (it cannot be 0 because this will contradict with the value of $q_{i}$), which implies that $q_{0}^* = 1 $. \\
if  $\sum_{j=1}^{i} \alpha_{j}^*(\Pi_{i+1j}^*-\Pi_{ij}^*) +\alpha_{i+1}^*\Pi_{i+1i+1}^* - \beta_0^*<0$ then $\gamma_{i}^* > 0$ (it cannot be 0 because this contradicts with the value of $q_{i}$), which implies that $q_{i}^* = 1 $\\

if  $\alpha_{j}^*\lambda_v \left(p_{i-1}{q_{i-1}}^* - p_{i-1} - p_{i}q_{i}^*\right) +\delta_i^*>0$ and we know that $\nu_{ij}^* \geq 0$ then  $\nu_{ij}^* = 0$ which gives $\Pi_{ij}^* \ne 1 $ and $\Pi_{ij}^* = 0 $.\\
if  $\alpha_{j}^*\lambda_v \left(p_{n-1}{q_{n-1}}^* - p_{n-1} - p_{0}q_{0}^*\right) +\delta_n^*>0$ which gives $\Pi_{nj}^* \ne 1 $ and $\Pi_{nj}^* = 0 $\\
if  $\alpha_{j}^*\lambda_v \left(p_{i-1}{q_{i-1}}^* - p_{i-1} - p_{i}q_{i}^*\right) +\delta_i^*<0$ then  $\nu_{ij}^* > 0$ (it cannot be 0 because this will contradict with the value of $\Pi_{ij}$), which implies that $\Pi_{ij} = 1 $. \\
if  $\alpha_{j}^*\lambda_v \left(p_{n-1}{q_{n-1}}^* - p_{n-1} - p_{0}q_{0}^*\right) +\delta_n^*<0$ then  $\nu_{nj}^* > 0$ (it cannot be 0 because this will contradict with the value of $\Pi_{nj}$), which implies that $\Pi_{nj} = 1 $. \\

We have also from the KKT conditions given by equation in in (\ref{eq:28}) that says either the Lagrangian coefficient is 0 or its the associated inequality is an equality: \\
if $\beta_1^* \ne 0 $ we have $q_{0}^* =\frac{\mu_c}{\lambda_v^* p_0}$\\
if $\alpha_{n}^* \ne 0$ we have $\frac{ \lambda_c^{(n)}+\lambda_v p_{n-1}q_{n-1}^*\Pi_{nn}^* - \lambda_v p_{n-1} \Pi_{nn}^* + R^*}{\lambda_v p_{0} \Pi{nn}^* } $\\
if $\alpha_{i}^* \ne 0$ , we have $ q_{i}^* =\frac{ R^*+\lambda_c^{(i)}+\lambda_v\left[\sum_{k=i+2}^{n-1}(p_{k-1}{q_{k-1}^*} -p_{k}{q_{k}^*})\Pi_{ki} + (p_{n-1}{q_{n-1}^*} -p_{0}{q_{0}^*})\Pi_{ni} - \sum_{k=i}^{n}p_{k-1}\Pi_{ki}+p_{i-1}q_{i-1}^*\Pi_{ii}^*-p_{i+1}q_{i+1}\Pi_{i+1i+1} \right]}{\lambda_v p_{i}(\Pi_{ii}-\Pi_{i+1i+1}) } $ \\for $i= 1, \ldots, n-1$\\

Otherwise by the Lagrangian relaxation: \\
$q_{i}^*=\zeta_i(R^*,q^*,\Pi^*,\alpha^*,\beta^*,\gamma^*,\omega^*,\mu^*,\nu^*,\delta^*)$ for $i= 1, \ldots, n-1$ and $\Pi_{ij}^*=\zeta_{ij}(R^*,q^*,\Pi^*,\alpha^*,\beta^*,\gamma^*,\omega^*,\mu^*,\nu^*,\delta^*)$ 
where $\zeta_i$ and $\zeta_{ij}$ are the solution that that maximize $\underset{\mathbf{q},\mathbf{\Pi}}\inf~ L(\mathbf{q},\Pi^*,R^*,\alpha^*,\beta^*,\gamma^*,\omega^*,\mu^*,\nu^*,\delta^*)$\\

Now in order to find the expression of $R^*$ we first look at its upper bound associated condition in (\ref{eq:28}). From there we can say that if 
$\omega_n^* \ne 0$ then $R^* = \epsilon_2 $ and if $\gamma_n^* \ne 0$ then $R* = \frac{\lambda_v - \sum ^{n}_{i=1}{\lambda_c^{(i)}}}{n} $ \\

Otherwise, from the last equation in (\ref{eq:30}),
if $\omega_n^* = 0$ and $\gamma_n^*=0$ then
\begin{equation}
\begin{aligned}
&R^* = \sum_{i=1}^{n-1}\alpha_{i}^*( \lambda_v \sum_{k=i}^{n-1} (p_{k-1} - p_{k-1}{q_{k-1}^*} + p_{k}q_{k}^*)\Pi_{ki}^*+ \lambda_v (p_{n-1} - p_{n-1}{q_{n-1}^*} + p_{0}q_{0}^*)\Pi_{ni}^* - \lambda_c^{(i)})\\ & + \alpha_{n}^* ( \lambda_v (p_{n-1} - p_{n-1}{q_{n-1}^*} + p_{0}q_{0}^*)\Pi_{nn}^* - \lambda^{(n)}_c  )\\
\end{aligned}
\end{equation}

\end{document}